\documentclass[final]{siamltex}

\newcommand{\be}{\begin{equation}}
\newcommand{\ee}{\end{equation}}
\newcommand{\ba}{\begin{array}}
\newcommand{\ea}{\end{array}}
\newcommand{\bea}{\begin{eqnarray}}
\newcommand{\eea}{\end{eqnarray}}
\newcommand{\calH}{{\cal H }}

\newcommand{\calS}{{\cal S }}
\newcommand{\calQ}{{\cal Q }}
\newcommand{\calC}{{\cal C }}

\newcommand{\calN}{{\cal N }}
\newcommand{\calZ}{{\cal Z }}
\newcommand{\la}{\langle}
\newcommand{\ra}{\rangle}
\newcommand{\prob}{{\mathrm Pr}}
\newcommand{\nn}{\nonumber}
\newcommand{\yes}{\mbox{\small \it yes}}
\newcommand{\no}{\mbox{\small \it no}}
\newcommand{\unsat}[1]{\mathrm{unsat}{(#1)}}
\newcommand{\trace}{\mathrm{Tr}}

\title{Complexity of stoquastic frustration-free  Hamiltonians}

\author{Sergey Bravyi\thanks{IBM  Watson Research Center, Yorktown Heights, NY 10598, USA.
{\tt sbravyi@us.ibm.com}} \and  Barbara Terhal\thanks{IBM  Watson
Research Center, Yorktown Heights, NY 10598, USA. {\tt
terhal@watson.ibm.com}}}

\begin{document}

\maketitle

\begin{abstract}
We study several problems related to properties of non-negative matrices that arise
at the boundary between quantum and classical probabilistic computation.
Our results are twofold.
First, we identify a large class of quantum Hamiltonians describing
systems of qubits for which the adiabatic evolution can be
efficiently simulated on  a classical probabilistic computer. These
are  stoquastic local Hamiltonians with a ``frustration free"
ground-state. A Hamiltonian belongs to this class iff it can be
represented as $H=\sum_a H_a$ where (1)  every term $H_a$ acts
non-trivially on a constant number of qubits,
 (2)  every term $H_a$ has real non-positive off-diagonal matrix elements
in the standard basis, and (3) the ground-state of $H$ is a
ground-state of every term $H_a$. Secondly, we  generalize the
Cook-Levin theorem proving NP-completeness of the satisfiability
problem to the complexity class MA  --- a probabilistic analogue of
NP. Specifically, we construct a quantum version of the $k$-SAT
problem which we call  ``stoquastic $k$-SAT" such that stoquastic
$k$-SAT is contained in MA for any constant $k$, and any promise
problem in MA is Karp-reducible to stoquastic $6$-SAT. This result
provides the first non-trivial example of a MA-complete promise
problem.
\end{abstract}



\pagestyle{myheadings}
\thispagestyle{plain}
\markboth{S. BRAVYI AND B. TERHAL}{Complexity of stoquastic frustration-free Hamiltonians}

\section{Introduction}
Recent years have seen the first steps in the development of a
quantum or a matrix-valued complexity theory. Such complexity theory
is interesting for a variety of reasons. Firstly, it increases our understanding of fundamental limitations imposed
on computational devices by the laws of physics.
Secondly, since quantum computation is an extension of classical computation, quantum
complexity theory provides a framework and a new angle
to attack major problems in classical complexity theory, see for instance~\cite{AaronsonPP}.

The problems for which a solution can be efficiently {\em found} on a quantum computer
constitute the class BQP --- the quantum analogue of the classical class BPP.
On the other hand, the problems for which a solution can be efficiently {\em verified}
on a quantum computer constitute the class QMA --- the quantum analogue of the class MA.
It was realized recently  that one can learn a lot about the classes
BQP and QMA by studying ground-state properties of {\em local Hamiltonians}.
Firstly, Aharonov et al.~\cite{ADKLLR:04} proved  that {\em
any} quantum algorithm can be executed via quantum  adiabatic evolution  in which parameters of a local Hamiltonian are changed adiabatically
while the state of the quantum computer is encoded into the
instantaneous ground-state.
Secondly, it was shown by Kitaev~\cite{KSV:computation} that
determining the ground-state energy of a local Hamiltonian is a
problem complete for the class QMA, see~\cite{KKR:06,OT:05,AGK:07}
for some recent progress.

In this paper we present a large family of local Hamiltonians for
which quantum adiabatic evolution can be efficiently simulated by a
classical probabilistic algorithm while determining the ground-state
energy is a problem complete for the class MA (considered as a class
of promise problems). To introduce this family of Hamiltonians let
us start from setting  up some terminology. The Hilbert space of $n$
qubits equipped with the standard basis $\{|x\ra\}$, $x\in
\Sigma^n$, will be denoted as $\calQ^n$. Here and below we denote as
$\Sigma^n=\{0,1\}^n$ the set of $n$-bit binary strings.

\begin{definition}
A $k$-local Hamiltonian acting on $n$ qubits is a Hermitian operator
$H$ on a Hilbert space $\calQ^n$ representable as $H=\sum_{a=1}^M
H_a$, where every term $H_a$ acts non-trivially only on some subset
of $k$ or less qubits. We shall consider families of $k$-local
Hamiltonians in which $k=O(1)$, $M\le poly(n)$, and $\|H_a\|\le
poly(n)$.
\end{definition}

We shall often use the terms $k$-local Hamiltonian and local Hamiltonian interchangeably.

\begin{definition}
A local Hamiltonian $H=\sum_{a} H_a$ is called frustration-free if
$H_a$ are positive semi-definite operators and the ground-state of
$H$ is a zero eigenvector of all operators  $H_a$.
\end{definition}

\begin{definition}
A local Hamiltonian $H=\sum_{a} H_a$ is called stoquastic with respect to a basis
$\cal B$ iff all $H_a$ have real non-positive off-diagonal matrix elements in the basis $\cal B$.
\end{definition}

Throughout this paper we shall consider Hamiltonians that are
stoquastic with respect to the standard basis of $n$ qubits formed
by tensor products of $|0\ra$ and $|1\ra$ states. The term
``stoquastic" was introduced in~\cite{BDOT:06} to emphasize a
connection with both stochastic matrices and quantum Hamiltonians.
The problem of determining the ground-state energy of a stoquastic
local Hamiltonian is contained in the complexity class AM
(Arthur-Merlin games), see~\cite{BDOT:06} for details.

In the present paper we shall focus on stoquastic frustration-free
(SFF) Hamiltonians. Some examples of SFF Hamiltonians will be given
in Section~\ref{sec:examples}.
These examples demonstrate that SFF Hamiltonians arise naturally at
the boundary between quantum computation and classical probabilistic
computation or, in physics, at the boundary between classical
statistical mechanics and quantum mechanics.

\subsection{Summary of results}

Our first result concerns the computational power of adiabatic
quantum evolution with SFF Hamiltonians. Let $H_{in}$ and $H_{f}$ be
SFF Hamiltonians acting on $n$ qubits. We assume that $H_{in}$ and
$H_{f}$ can be connected by an adiabatic path $H(s)$, $0\le s\le 1$,
such that $H_{in}=H(0)$, $H_{f}=H(1)$.
In addition the following conditions should be met

\noindent {\bf
(A0)} Hamiltonians $H(s)$ are stoquastic and frustration-free for
all $s$.

\noindent {\bf (A1)} The path is sufficiently smooth; $J=\max_s \|
d{H(s)} / ds\|\le poly(n)$.

\noindent {\bf (A2)} The Hamiltonian $H(s)$ has a non-degenerate
ground-state for all $s$. The spectral gap $\Delta(s)$  between the
smallest and the second smallest eigenvalues of $H(s)$ is
sufficiently large: $\Delta=\min_s \Delta(s) \ge 1/poly(n)$.

\noindent {\bf (A3)} The initial Hamiltonian is sufficiently simple
so there exists a $poly(n)$ algorithm that finds a basis vector
$|x\ra$ such that the overlap between $|x\ra$ and the ground-state
of $H_{in}$ is at least $2^{-poly(n)}$.

In contrast to the standard paradigm of adiabatic quantum
computation or quantum annealing we do not require  the adiabatic
path to be  a linear interpolation between $H_{in}$ and $H_f$ (since
otherwise it may be impossible to fulfill the frustration-free
condition). The goal of the simulation is to sample $x\in \{0,1\}^n$
from the probability distribution $\pi(x)$ associated with the
ground-state $|\psi\ra$ of $H_{f}$, that is, $\pi(x)=|\la
x|\psi\ra|^2$ (assuming that $\la \psi|\psi\ra=1$).  Our first
result is as follows.
\begin{theorem}
\label{thm:adiabatic} Let $|\psi\ra$ be the ground-state of $H_{f}$
and $\pi(x)=|\la x|\psi\ra|^2$. Suppose the adiabatic evolution
conditions (A0)-(A3) are met. Then for any precision $\delta>0$
there exists a classical probabilistic algorithm that generates a
random variable $x\in \{0,1\}^n$ with a probability distribution
$\tilde{\pi}(x)$ such that $\|\tilde{\pi}-\pi\|_1\le \delta$. The
running time of the algorithm is $poly(n,\delta^{-1})$.
 \end{theorem}

It should be mentioned that  the property of being frustration-free
alone cannot render the adiabatic evolution efficiently simulatable
on a classical computer.
The following proposition is a simple extension of~\cite{ADKLLR:04}.
\begin{proposition}
\label{prop:adiabatic1}
Let $U$ be a quantum circuit with $n$ input qubits initialized in
the state $|0\ra$ and $L=poly(n)$ two-qubit gates. Let $|\psi_L\ra$
be the $n$-qubit output state of $U$. For any precision $\delta>0$
one can construct a family of frustration-free Hamiltonians $H(s)$,
$0\le s\le 1$ acting on $poly(n,\delta^{-1})$ qubits and satisfying
conditions (A1),(A2),(A3) such that the ground-state of $H(1)$
approximates $|\psi_L\ra$ with precision $\delta$ (after discarding
some ancillary qubits).
\end{proposition}

For the sake of completeness we outline the proof in Appendix~B.
 Proposition~\ref{prop:adiabatic1} implies that simulating the adiabatic evolution
 with frustration-free Hamiltonians is as hard as simulating a universal quantum computer.
On the other hand, our technique does not permit an efficient
simulation of an adiabatic evolution with stoquastic Hamiltonians
which are not necessarily frustration-free. In fact, adiabatic
evolution with stoquastic Hamiltonians includes a variety of the
quantum annealing algorithms, see~\cite{Farhi:00}.
It is an interesting open question whether adiabatic evolution with
general stoquastic Hamiltonians can be simulated classically in
polynomial (or at least sub-exponential) time.
Our result raises a question: how hard is it to verify that a given
local Hamiltonian $H=\sum_a H_a$ is SFF? Clearly, the property of
being stoquastic can be verified efficiently since every term $H_a$
acts on $O(1)$ qubits. On the other hand, verifying that $H$ is
frustration-free requires evaluating the smallest eigenvalue of $H$.
This problem is known to be QMA-hard for general local Hamiltonians,
see~\cite{KSV:computation}. For stoquastic Hamiltonians the smallest
eigenvalue problem is known to be contained in QMA~$\cap$~AM, where
AM  (Arthur-Merlin games) is a probabilistic analogue of NP in which
the prover and the verified can exchange a constant number of
messages, see~\cite{BDOT:06}.
One can expect that verifying whether a stoquastic Hamiltonian is
frustration-free must be easier than evaluating the smallest
eigenvalue because for a positive instance we have the additional
information that the ground-state minimizes the expectation of every
local term in a Hamiltonian. Let us start from stating the problem
more formally.

\begin{definition}
A system of $(n,k)$-constraints is a family of $n$-qubit positive semidefinite
Hermitian operators $\{H_a\}$, $a=1,\ldots,M$, such that every $H_a$ acts non-trivially only on some subset of $k$
or less qubits. A system $\calC=\{H_a\}$ is satisfiable iff there exists a state $|\psi\ra$ such that
$H_a\, |\psi\ra=0$ for all $H_a$.
Such a state is called a satisfying assignment.
 We shall consider systems in which $k=O(1)$, $M\le poly(n)$ and $\|H_a\|\le poly(n)$.
\end{definition}

If a system $\calC=\{H_a\}$ is not satisfiable, any state $|\psi\ra$
violates at least one constraint, i.e., $\la \psi|H_a|\psi\ra>0$ for
some $H_a$. Let us define the {\it unsat-value} of a system $\calC$
as the smallest eigenvalue of a Hamiltonian $H_{\calC}=\sum_a H_a$,
\[
\unsat{\calC}=\min_{\psi\, : \, \la \psi|\psi\ra=1} \la\psi|H_{\calC} |\psi\ra, \quad H_{\calC}=\sum_a H_a.
\]
By definition, a system $\calC$ is satisfiable iff the corresponding Hamiltonian $H_{\calC}$
is frustration-free.
The quantum $k$-SAT problem is to distinguish the case when a system of $(n,k)$-constraints is
satisfiable from the case when it has a non-negligible (i.e. polynomial in $1/n$) unsat-value.
The quantum version
of the Cook-Levin theorem proved  by Kitaev~\cite{KSV:computation}
and developed further in~\cite{Bravyi06,Nagaj06} asserts that
quantum $k$-SAT belongs to QMA for any constant $k$ and
the quantum $4$-SAT is complete
for the class QMA${}_1$ (the analogue of QMA with zero completeness error).
Let us now define a {\it stoquastic} systems of constraints and the problem stoquastic $k$-SAT.
\begin{definition}
A system of $(n,k)$-constraints $\calC=\{H_a\}$ is called stoquastic iff every $H_a$
has real non-positive off-diagonal matrix elements in the standard basis.
\end{definition}

\begin{definition}
An instance of stoquastic $k$-SAT is a tuple $(n,\calC,\epsilon)$,
where  $\calC$ is a stoquastic system of $(n,k)$-constraints  and $\epsilon=n^{-O(1)}$ is a positive number. For yes-instances $\unsat{\calC}=0$. For no-instances $\unsat{\calC}\ge \epsilon$.
\end{definition}

It should be emphasized that stoquastic  $k$-SAT is a promise problem. It can be
represented by a pair of languages $(L_{\yes},L_{\no})$ such that
\[
L_{\yes}=\{ (n,\calC,\epsilon)\, : \, \unsat{\calC}=0\}, \quad L_{\no}=\{(n,\calC,\epsilon)\, : \, \unsat{\calC}\ge \epsilon\}.
\]
Note that classical $k$-SAT can be obtained as a special case of
stoquastic $k$-SAT when all the constraints $H_a$ are diagonal in
the standard basis with matrix elements $0,1$ on the diagonal. It
follows that  stoquastic $k$-SAT is NP-hard for $k\ge 3$. Our second
result is that  stoquastic $k$-SAT can be placed in the complexity
class MA --- a probabilistic analogue of NP with only one message
sent from  the prover to the verifier. For the sake of completeness
we present  a formal definition of MA and Promise-MA in Appendix~A.
In addition, we proved in \cite{BDOT:06} that stoquastic $k$-SAT is
complete for the class Promise-MA for sufficiently large $k$.
Putting these results together gives

\begin{theorem}\label{thm:MA}
The promise problem stoquastic $k$-SAT is contained in MA for any constant $k$.
Any promise problem in MA is Karp-reducible to
stoquastic $6$-SAT with constraints
$\{H_a=I-\Pi_a\}$ where $\Pi_a$ are projectors with matrix elements from a set  $\{ 0,\frac12,1\}$.
\end{theorem}

This result can be regarded as a generalization of the Cook-Levin theorem proving NP-completeness
of the classical satisfiability problem.


The rest of the paper is organized as follows. In
Section~\ref{sec:pwork} we briefly review the previous work on the
subject. Section~\ref{sec:techniques} sketches the main ideas and
techniques used in the proof of Theorem~\ref{thm:adiabatic} and
Theorem~\ref{thm:MA}. Section~\ref{sec:examples} provides some
interesting examples of SFF Hamiltonians. Basic properties of
non-negative matrices required for understanding of our simulation
algorithms are presented in Section~\ref{sec:nonnegative}. A proof
of Theorem~\ref{thm:adiabatic} can be found in
Section~\ref{sec:adiabatic}. Theorem~\ref{thm:MA} is proved in
Section~\ref{sec:inMA}. In Section~\ref{sec:discussion} we discuss
some open problems and directions for future work. A formal
definition of the classes MA and Promise-MA is given in Appendix~A.
Finally, Appendix~B contains a proof of
Proposition~\ref{prop:adiabatic1} and a proof of MA-hardness of
stoquastic $6$-SAT.

\subsection{Previous work} \label{sec:pwork}

Stoquastic Hamiltonians are well known in computational physics as
Hamiltonians avoiding the ``sign problem". It was realized decades
ago that ground-state properties of such Hamiltonians can be
simulated using classical Monte Carlo algorithms,
see~\cite{Ceperley:90,Buonaura:98}, for systems as large as several
hundred qubits.
The general limitations of such algorithms which are likely to make
them inefficient in the complexity-theoretic sense were also well
understood, see e.g.~\cite{Buonaura:98,Hetherington:84}. The first
rigorous attempt to analyze the complexity of the smallest
eigenvalue problem for stoquastic Hamiltonians was made
in~\cite{BDOT:06}. It was shown that this problem belongs to the
complexity class AM (Arthur-Merlin games). Using the same ideas the
smallest eigenvalue problem was shown in the
unpublished~\cite{BBT:06} to be contained in a smaller class
SBP~$\subseteq$~AM, where SBP stands for Small Bounded-Error
Probability, see~\cite{BGM:03}. The complexity of the smallest
eigenvalue problem for $k$-local stoquastic Hamiltonians was shown
to be the same for all $k\ge 2$, see~\cite{BDOT:06}. A related
problem called ``Stoquastic Local Consistency" which involves
verifying certain consistency conditions for a collection of local
density matrices was studied by Liu~\cite{Liu:07}.

A connection between stoquastic Hamiltonians and classical
probabilistic computation was studied by Aharonov and Ta-Shma
in~\cite{Aharonov:03}. Using the technique of adiabatic state
generation these authors constructed quantum algorithms for
q-sampling from the stationary distribution of a reversible Markov
chain satisfying certain additional properties. An analogous
connection between stoquastic Hamiltonians and classical statistical
mechanics was obtained by Verstraete et al.~\cite{VWPC:06} and Somma
et al.~\cite{SBO:06}. These authors proved that a coherent version
of the Gibbs thermal state associated with any local classical
Hamiltonian can be represented as the unique ground-state of
stoquastic frustration-free Hamiltonian.

\section{Techniques}
\label{sec:techniques} This section highlights the main ideas and
techniques used  in the rigorous proofs in
Sections~\ref{sec:adiabatic} and~\ref{sec:inMA}.

\subsection{A random walk associated with a SFF Hamiltonian}
The main technical tool used throughout the paper is a novel random
walk algorithm that allows one to simulate some ground-state
properties of SFF Hamiltonians.
This algorithm is similar in spirit to the Green Function Monte
Carlo method (GFMC)  --- a probabilistic heuristic for the
simulation of quantum spin systems,
see~\cite{Ceperley:90,Buonaura:98}. However, in contrast to GFMC our
algorithm offers rigorous upper bounds on the running time and the
error probability.

Let $H=\sum_a H_a$ be some SFF Hamiltonian and $|\psi\ra$ be a
ground-state of $H$, i.e, $H_a\,|\psi\ra=0$ for all $a$. Using the
Perron-Frobenius theorem we will show in
Section~\ref{sec:nonnegative} that a ground-state of any SFF
Hamiltonian can be chosen as a vector with real non-negative
amplitudes in the standard basis (if the smallest eigenvalue has
multiplicity $q$, one can choose $q$ orthonormal ground-states such
that each state has non-negative amplitudes). For that reason we can
assume that $\la x|\psi\ra\ge 0$ for all $x\in \Sigma^n$. A set of
binary strings that appear in $|\psi\ra$ with a non-zero amplitude
will be called a {\it support} of $|\psi\ra$ and denoted as \be
\label{S(psi)} \calS(\psi)=\{ x\in \Sigma^n \, : \, \la x|\psi\ra
>0\}. \ee

A random walk associated with a Hamiltonian $H$ and a ground-state
$|\psi\ra$ is a random walk on the set $\calS(\psi)$ with transition
matrix \be \label{transition|matrix} P_{x\to y} = \frac{\la
y|\psi\ra}{\la x|\psi\ra} \, \la y|G|x\ra, \quad G=I-\beta\, H,
\quad \mbox{for any $x\in \calS(\psi)$}. \ee Here $\beta>0$ is a
real parameter that is chosen sufficiently small in order to make
$G$ a matrix with non-negative entries, so that $P_{x\to y}\ge 0$.
One can infer directly from the definition that $P_{x\to y}=0$
unless $y\in \calS(\psi)$. Besides, the eigenvalue equation $G\,
|\psi\ra =|\psi\ra$ implies $\sum_{y\in S(\psi)} P_{x\to y}=1$ for
any $x\in \calS(\psi)$. Therefore Eq.~(\ref{transition|matrix})
indeed defines a random walk on the set $\calS(\psi)$. Direct
inspection shows that the stationary distribution of the random walk
is \be \label{fixed|point} \pi(x)=\la x|\psi\ra^2. \ee Here we
assumed that $|\psi\ra$ is a normalized state, i.e., $\la
\psi|\psi\ra=1$.
Note that for a given Hamiltonian $H$ one may have several
stationary distributions supported on mutually disjoint sets of
basis vectors such that each distribution is associated with some
non-negative ground-state of $H$.

Let us now argue that  the random walk defined in
Eq.~(\ref{transition|matrix}) can be efficiently simulated on a
classical probabilistic computer. Indeed, let $\Pi_a$ be the
spectral  projector corresponding to  the zero eigenvalue of $H_a$.
By definition, $\Pi_a\, |\psi\ra=|\psi\ra$ for all $a$. The crucial
property of the projectors $\Pi_a$ is that they have real
non-negative matrix elements in the standard basis. This property
can be proved using elementary algebra, see
Section~\ref{sec:nonnegative}. Moreover, we will show that any
projector with non-negative entries can be decomposed into  a direct
sum of rank-one projectors with non-negative entries. Using this
decomposition we shall be able to show that if $x\in \calS(\psi)$
and $\la y|H_a|x\ra <0$ for some $a$ and some $y\in \Sigma^n$ then
$y\in \calS(\psi)$ and \be \label{amplitude|ratio} \frac{\la
y|\psi\ra}{\la x|\psi\ra}= \sqrt{\frac{\la y|\Pi_a|y\ra}{\la
x|\Pi_a|x\ra}}. \ee
Note that $P_{x\to y}=0$ unless $\la
y|G|x\ra>0$, that is, $\la y|H_a|x\ra<0$ for some $a$. Since every
$H_a$ acts non-trivially on $O(1)$ qubits, the number of strings $y$
such that $P_{x\to y}>0$ is at most $poly(n)$.
 Therefore, given a current position of the walk $x\in S(\psi)$
one can efficiently simulate one step of the walk by first finding all strings $y$ such that
$\la y|H_a|x\ra<0$ for some $a$ and then using Eq.~(\ref{amplitude|ratio}) and Eq.~(\ref{transition|matrix})
to compute the transition probabilities.

\subsection{Simulation of the adiabatic evolution}
\label{subs:technique|adiabatic}
When we discretize the adiabatic evolution we get a family of SFF
Hamiltonians $H^{(j)}=\sum_{a=1}^M H_a(j/T)$, where $j=0,\ldots,T$
is a discrete time step. The Hamiltonian $H^{(j)}$ has a unique
non-negative ground-state $|\psi^{(j)}\ra$ with support
$\calS(\psi^{(j)})$. Using definition Eq.~(\ref{transition|matrix})
gives us a family of efficiently simulatable random walks
$P^{(0)},P^{(1)},\ldots,P^{(T)}$ such that $P^{(j)}$ is a random
walk on the set $\calS(\psi^{(j)})$. The walk $P^{(j)}$ has a
stationary distribution $\pi^{(j)}$ such that $\pi_x^{(j)}=\la
x|\psi^{(j)}\ra^2$. Note that the spectrum of $P^{(j)}$ coincides
with the spectrum of $G^{(j)}=I-\beta H^{(j)}$ restricted to a
subspace spanned by basis vectors from $S(\psi^{(j)})$. Since the
largest eigenvector of $G^{(j)}$ belongs to this subspace, the
spectral gap of $P^{(j)}$ (i.e. the gap between the largest and the
second largest eigenvalues) is at least the spectral gap of
$G^{(j)}$. Condition (A2) implies that the spectral gap of $G^{(j)}$
is at least $\beta \Delta$. Therefore the spectral gap of $P^{(j)}$
is at least $\beta \Delta$ which is polynomial in $1/n$. 

Recall that the goal of the simulation is to sample $x$ from the final distribution $\pi_x^{(T)}$.
Since the  spectral gap of the walk $P^{(T)}$ is polynomial in $1/n$,
all we need is a {\it warm start} for $P^{(T)}$, that is a string
$x\in S(\psi^{(T)})$ such that the stationary distribution
$\pi^{(T)}$ has a non-negligible probability at $x$. Here
non-negligible means $\pi^{(T)}_x\ge 2^{-poly(n)}$.
 Our strategy is to generate a warm start for
the walk $P^{(j+1)}$ using a warm start for $P^{(j)}$ by making sufficiently many steps
of the walk $P^{(j)}$ such that the endpoint of $P^{(j)}$ is a string sampled from the
stationary distribution $\pi^{(j)}$ (with an exponentially small error).
As far as the initial walk $P^{(0)}$ is concerned, a warm start can be efficiently generated due to condition (A3).
The main technical challenge is to bound the probability  of failure, i.e., the probability that for some $j$
the end-point of the walk $P^{(j)}$ is not a warm start for the next walk $P^{(j+1)}$.

In order to achieve this we introduce the notion of {\it
$t$-balanced strings}. Namely, given probability distributions $\pi$
and $\rho$ on the set $\Sigma^n$, a string $x\in \Sigma^n$ is
$t$-balanced with respect to $\pi$ and $\rho$ iff $\pi_x,\rho_x>0$
and $t^{-1}\le \pi_x/\rho_x\le t$. We show that for sufficiently
large (but constant)  $t$ the probability for a string $x$ drawn
from $\pi$ to be $t$-balanced is at least $1- O(1-F(\pi,\rho))$,
where $F(\pi,\rho)=\sum_x \sqrt{\pi_x \rho_x}$ is the fidelity
between $\pi$ and $\rho$. Using conditions (A1) and (A2) we shall
bound the fidelity between $\pi^{(j)}$ and $\pi^{(j+1)}$ as \be
\label{fidelity|bound} F(\pi^{(j)},\pi^{(j+1)})\ge 1- \frac{J^2}{T^2
\Delta^2}. \ee It follows that a string drawn from $\pi^{(j)}$ is
$t$-balanced (for some $t=O(1)$) with respect to $\pi^{(j)}$ and
$\pi^{(j+1)}$ with probability at least $1-O(J^2 T^{-2}
\Delta^{-2})$.
Choosing $T\gg J^2 \Delta^{-2}$ and choosing the number of steps in
each walk much larger than the inverse spectral gap we shall prove that the
end-point of the walk $P^{(j)}$ is a warm start for the walk
$P^{(j+1)}$ for all $j=0,\ldots,T-1$ with probability at least
$1-O(J^2 T^{-1} \Delta^{-2}) \approx 1$.

\subsection{Stoquastic $k$-SAT is contained in MA}
\label{sec:pre_inMA} Let $(n,\calC,\epsilon)$ be an instance of the
stoquastic $k$-SAT problem. Here $\calC=\{H_a\}_{a=1,\ldots,M}$ is a
system of stoquastic constraints. Define a Hamiltonian
$H=\sum_{a=1}^M H_a$. By definition,
\[
\ba{rcl}
\mbox{$(n,\calC,\epsilon)$ is yes-instance} &\Rightarrow& \mbox{$H$ is a SFF Hamiltonian}, \nn \\
\mbox{$(n,\calC,\epsilon)$ is no-instance} &\Rightarrow& \mbox{The smallest eigenvalue of $H$ is at least $\epsilon$.} \nn
\ea
\]
Let $\Pi_a$ be the spectral  projector corresponding to the zero
eigenvalue of $H_a$. We shall partition the set of all binary
strings $\Sigma^n$ into good and bad strings, such that \be
\label{good|bad} S_{good}=\{ x\in \Sigma^n \, : \, \la x|\Pi_a|x\ra
>0 \quad \mbox{for all $a=1,\ldots,M$} \} \ee and $S_{bad}=\Sigma^n
\backslash S_{good}$. For any instance $(n,\calC,\epsilon)$ and any
$x\in S_{good}$ define the transition probabilities \be
\label{sat|transition|matrix} P_{x\to y} =\frac1M \sum_{a=1}^M
\sqrt{\frac{\la y|\Pi_a|y\ra}{\la x|\Pi_a|x\ra}} \, \la y |G_a
|x\ra, \quad G_a = I-\beta H_a, \ee where $\beta>0$ is a real
parameter  chosen sufficiently small so that all matrices $G_a$ are
non-negative. One can infer directly from the definition that
$P_{x\to y}=0$ unless $\la y|\Pi_a|y\ra>0$ for some $a$. Therefore,
for any given $x$ the number of strings $y$ such that $P_{x\to y}>0$
is at most $poly(n)$.
The property that $\Pi_a$ is a direct sum of
rank-one projectors with non-negative entries implies that the
transition probabilities are normalized, that is, $\sum_{y\in
\Sigma^n} P_{x\to y}=1$ for any $x\in S_{good}$, see
Section~\ref{sec:nonnegative} for details. However, the transition
probabilities Eq.~(\ref{sat|transition|matrix}) do not automatically
define a random walk on the set of good strings because one may have
transitions from a good string to a bad string. In other words,
Eq.~(\ref{sat|transition|matrix}) permits transitions $S_{good}\to
S_{good}$ as well as $S_{good}\to S_{bad}$.

It turns out that if $(n,\calC,\epsilon)$  is a positive instance
then Eq.~(\ref{sat|transition|matrix}) does define a random walk on
some subset of good strings. Indeed, let $|\psi\ra$ be a {\it
satisfying assignment}, i.e., a state satisfying $H_a\, |\psi\ra=0$
for all $a$. As was mentioned above, we can assume that $|\psi\ra$
has real non-negative amplitudes. Using the eigenvalue equations
$\Pi_a\, |\psi\ra=|\psi\ra$ one can easily show that $|\psi\ra$ is
supported only on good strings, $S(\psi)\subseteq S_{good}$, and for
any $x\in S(\psi)$ the transition probabilities
Eq.~(\ref{sat|transition|matrix}) can be expressed as \be
\label{sat|transition|matrix1} P_{x\to y} = \frac{\la y|\psi\ra}{\la
x|\psi\ra} \la y|G|x\ra, \quad G=\frac1M\sum_{a=1}^M G_a. \ee
Repeating the same arguments as in
Section~\ref{subs:technique|adiabatic} we conclude that
Eq.~(\ref{sat|transition|matrix1}) is a transition matrix of a
random walk on the set $S(\psi)$. For any given starting string
$x\in S(\psi)$ the walk stays in $S(\psi)$ forever. The walk can be
efficiently simulated on a classical probabilistic computer. (In
contrast to Section~\ref{subs:technique|adiabatic} the mixing time
of the walk is not a matter of concern.)

Suppose now that $(n,\calC,\epsilon)$ is a negative instance. One
can still use Eq.~(\ref{sat|transition|matrix}) to simulate a random
walk starting from a good string until the first time the walk hits
a bad string. It will be shown in Section~\ref{sec:inMA} that the
probability for the walk starting from a good string to stay in
$S_{good}$ for $L$ steps (and satisfy some extra tests which are
always passed for positive instances) decays approximately as
$(1-\beta \epsilon/M)^L$. Given the polynomial bounds on $\beta$,
$\epsilon$, and $M$ we can make this probability exponentially small
with $L=poly(n)$.

In order to prove that $(n,\calC,\epsilon)$ is a positive instance, the prover can send the verifier
a binary string $w\in \Sigma^n$ such that some satisfying assignment $|\psi\ra$
has a non-negligible amplitude on $w$ (for a negative instance $w$ can be arbitrary string).
The verifier checks whether $w\in S_{good}$ and simulates $L=poly(n)$ steps of the random walk
defined above starting from $w$. Whenever the walk hits a bad string, the verifier aborts the simulation and
outputs 'no'. If the walk stays in $S_{good}$ for $L$ steps, the verifier outputs 'yes' (conditioned on the
outcome of some extra tests described in Section~\ref{sec:inMA}).

\subsection{MA-hardness}
MA-hardness of the stoquastic $6$-SAT problem with the constraints
satisfying conditions of Theorem~\ref{thm:MA} follows directly from
from MA-hardness of the stoquastic $6$-local Hamiltonian problem,
see Lemma~3 in~\cite{BDOT:06}. In order to make the paper
self-contained we repeat the proof of~\cite{BDOT:06} with some minor
modifications in Appendix~B.

\section{Examples of stoquastic frustration-free Hamiltonians}
\label{sec:examples}

In this section we give some examples of SFF Hamiltonians. Firstly,
the results obtained by Verstraete et al.~\cite{VWPC:06} and Somma
et al.~\cite{SBO:06} imply that a coherent version of the thermal
Gibbs state associated with any {\em classical} local Hamiltonian
can be represented as the unique ground-state of a SFF Hamiltonian.
Secondly, we use the clock Hamiltonian construction
from~\cite{KSV:computation} to show that a coherent version of a
probability distribution generated by any polynomial-size classical
reversible circuit can be approximated by the ground-state of a SFF
Hamiltonian.

\subsection{Coherent thermal states of classical Hamiltonians}
Let $H$ be any classical Hamiltonian acting on $n$ qubits ($H$ has a
diagonal matrix in the standard basis). Denote $H(x)=\la x|H|x\ra$.
Choose any $\beta>0$ and  consider a coherent version of the thermal
Gibbs state \be \label{Gibbs} |\pi\ra=\calZ^{-1/2} \sum_{x\in
\Sigma^n} e^{-\beta H(x)/2 }\, |x\ra, \quad \calZ=\trace{\,e^{-\beta
H}}. \ee Given a state $|\pi\ra$ one can sample $x$ from the Gibbs
distribution $\pi_x = \calZ^{-1} e^{-\beta    H(x)}$ by measuring
the state in the standard basis. More interestingly, given coherent
Gibbs states $|\pi\ra$ and $|\pi'\ra$ corresponding to some
classical Hamiltonians $H$ and $H'$, one can perform the swap test
on the two states thus evaluating the statistical difference between
the two Gibbs distributions, see~\cite{Aharonov:03}.

Suppose that $H$ is a local Hamiltonian, $H=\sum_a H_a$, and each
qubit is acted on by a constant number of terms $H_a$. The analysis
performed in~\cite{VWPC:06,SBO:06} shows  that the state
Eq.~(\ref{Gibbs}) is the unique ground-state of some SFF Hamiltonian
$H_\beta$.  Indeed, one can represent $|\pi\ra$ as follows (ignoring
the normalization)
\[
|\pi\ra=e^{-\beta H/2} \, |+\ra, \quad |+\ra = \sum_{x\in \Sigma^n} |x\ra.
\]
Let $X_j$ be the Pauli $X$-matrix, $X=\left(\ba{cc} 0 &1 \\ 1 & 0\\ \ea\right)$,
 acting on qubit $j$. Using the representation above one can easily check that
\[
X_j\, |\pi\ra = \Gamma_j \, |\pi\ra, \quad \Gamma_j =X_je^{-\beta H/2 } X_j e^{\beta H/2}
\]
for all $j=1,\dots,n$. Note that the operator $\Gamma_j$ is diagonal in the standard basis.
Since all matrix elements of $\Gamma_j$ are real, we conclude that $\Gamma_j$ is Hermitian.
Since we assumed that any qubit is acted on by a constant number of local terms in $H$, we conclude that
$\Gamma_j$ acts non-trivially on a constant number of qubits.
Define a Hamiltonian
\[
H_\beta=\sum_{j=1}^n \Gamma_j - X_j
\]
such that $H_\beta \, |\pi\ra=0$. The Perron-Frobenius theorem
implies that $|\pi\ra$ is the unique ground-state of $H_\beta$
(indeed, $G=I-\gamma H$ is a non-negative irreducible matrix for
sufficiently small $\gamma>0$ and $|\psi\ra$ is a positive
eigenvector of $G$). The same argument shows that $|\pi\ra$ is a
ground-state of every local term $\Gamma_j-X_j$.  It follows that
$\Gamma_j-X_j$ is a positive semi-definite operator and thus
$H_\beta$ is a SFF Hamiltonian with unique ground-state $|\pi\ra$.

\subsection{Coherent probabilistic computation}
\label{subs:coherent}
Let $U$ be a classical polynomial-size circuit with reversible gates
(e.g. Toffoli gates) with $n$
input and $n$ output bits. Assume that the first $k$ input bits are
drawn from the uniform distribution and the last $n-k$ input bits
are initialized to $0$. Let  $\pi$ be the corresponding  probability
distribution  of the output bits. Consider a coherent version of
$\pi$, \be \label{pi|state} |\pi\ra=\sum_{x\in \Sigma^n}
\sqrt{\pi_x} \, |x\ra. \ee
We claim that  $|\pi\ra$ can be
represented as the unique ground-state of some SFF Hamiltonian.
More strictly, for any precision $\delta>0$ there exists a SFF
Hamiltonian $H$ acting on  $poly(n,\delta^{-1})$ qubits such that $H$ has an unique ground-state
$|\psi\ra$ satisfying $\la \psi|\pi\otimes \phi_a\ra\ge 1-\delta$
for some simple ancillary state $|\phi_a\ra$.
Such a Hamiltonian $H$ can be constructed by transforming $U$ into a quantum circuit $\tilde{U}$
taking as input a state $|+\ra^{\otimes k}|0\ra^{\otimes (n-k)}$, where $|+\ra=(|0\ra+|1\ra)/\sqrt{2}$.
The circuit $\tilde{U}$ first applies the gates of $U$ in a coherent fashion and then
applies $\delta^{-1} L$ identity gates, where $L$ is the number of gates in $U$.
Note that the output state of $\tilde{U}$ is $|\pi\ra$. Applying the clock Hamiltonian construction of~\cite{KSV:computation}
to $\tilde{U}$  we get the desired Hamiltonian $H$.
The details of this constructions are presented in Appendix~B.

\section{Non-negative matrices: basic properties}
\label{sec:nonnegative} This section summarizes some basis
properties of non-negative matrices that are needed in understanding
our simulation algorithms. Let us start from setting up some
terminology and notations. A matrix is called non-negative iff all
its entries are real and non-negative. A term {\it non-negative
projector} will refer to a Hermitian projector acting on $\calQ^n$
which has a non-negative matrix in the standard basis. Analogously,
a term {\it non-negative state} will refer to a normalized vector
$|\psi\ra\in \calQ^n$ such that all amplitudes of $|\psi\ra$ in the
standard basis are real and non-negative. Let us start from a simple
observation.
\begin{proposition}
\label{prop:SFF}
Let $H$ be a Hermitian operator with non-positive off-diagonal matrix elements
in the standard basis. Then the spectral projector  $\Pi$
corresponding to the smallest eigenvalue of $H$ is non-negative.
\end{proposition}
\begin{proof}
Indeed, $\Pi=q \cdot \lim_{\beta\to \infty} e^{-\beta H}/Z(\beta)$,
where $Z(\beta)=\trace{(e^{-\beta H})}$
and $q$ is the multiplicity of the smallest eigenvalue.
 Since $e^{-\beta H}$ is a non-negative matrix
for any $\beta\ge 0$, the limit $\Pi$
is a non-negative matrix.
\end{proof}


Next we shall give a simple characterization of non-negative projectors.
\begin{lemma}
\label{lemma:projectors}
For any non-negative projector $\Pi$ of rank $q$ there exist
non-negative states $|\psi_1\ra,\ldots,|\psi_q\ra$ such that
$\la \psi_a|\psi_b\ra=\delta_{a,b}$ for all $a, b$ and
$\Pi=\sum_{a=1}^q |\psi_a\ra\la \psi_a|$.
\end{lemma}

Note that non-negative states  are pairwise orthogonal iff they have
support on non-overlapping subsets of basis vectors,
that is, $\calS(\psi_a)\cap \calS(\psi_b)=0$ for $a\ne b$. Thus the
lemma asserts that a non-negative projector is always
block-diagonal (up to a permutation of basis vectors) with each block
being a projector onto a non-negative state.
Combining Lemma~\ref{lemma:projectors} and Proposition~\ref{prop:SFF}
one concludes that
\begin{corollary}
\label{cor:ground|states} The ground-subspace of any stoquastic
Hamiltonian has an orthonormal basis of non-negative ground-states.
\end{corollary}

Let us proceed with the proof of Lemma~\ref{lemma:projectors}.
\begin{proof}
For any basis vector $|x\ra$
define a ``connected component" \[ T_x=\{y\in \Sigma^n \, : \,  \la x|\Pi|y\ra
>0\}. \] (Some of the sets $T_x$ may be empty.) For any triple
$x,y,z$ the inequalities $\la x|\Pi |y\ra >0$, $\la y|\Pi |z\ra >0$
imply $\la x|\Pi |z\ra >0$ since \[ \la x|\Pi |z\ra = \la x|\Pi^2
|z\ra =\sum_{u\in \Sigma^n} \la x|\Pi|u\ra \la u |\Pi|z\ra \ge \la
x|\Pi |y\ra \la y|\Pi |z\ra >0. \] A similar argument shows that if
$T_x$ is non-empty then  $\la x|\Pi|x\ra>0$, that is, $x\in T_x$.
Therefore the property $\la x|\Pi |y\ra>0$ defines a symmetric
transitive relation on $\Sigma^n$ and we have
\begin{itemize}
\item $y\in T_x$ implies $T_y=T_x$,
\item $y\notin T_x$ implies $T_y\cap T_x=\emptyset$.
\end{itemize}
Consider a subspace $\calH(T_x)\subseteq \calQ^n$ spanned by the basis
vectors from $T_x$. Clearly $\calH(T_x)$ is $\Pi$-invariant. Thus
$\Pi$ is block diagonal w.r.t. decomposition of the whole Hilbert
space into the direct sum of spaces $\calH(T_x)$
and the orthogonal complement where $\Pi$ is zero.
 Moreover, the
restriction of $\Pi$ onto any non-zero subspace $\calH(T_x)$ is a
projector with strictly positive entries. According to the
Perron-Frobenius theorem, the largest eigenvalue of a Hermitian
operator with positive entries is non-degenerate. Thus each block of
$\Pi$ has rank $1$, since a projector has eigenvalues $0$ and $1$
only.
\end{proof}

We shall use this characterization of non-negative projectors to
derive the following lemma that plays a key role in the definition
of a random walk in Section~\ref{sec:inMA}.

\begin{lemma}\label{lemma:key}
Let $\Pi$ be a non-negative projector.
Suppose $\Pi |\psi\ra = |\psi\ra$ for some non-negative state
$|\psi\ra$.
Then for any $x\in \calS(\psi)$ one has\\ \\
(1) $\la x|\Pi|x\ra >0$,\\ \\
(2) If $\la x|\Pi|y\ra>0$ for some $y\in \Sigma^n$ then
$y\in \calS(\psi)$ and
\be\label{ratio}
\frac{\la y|\psi\ra}{\la x|\psi\ra} = \sqrt{\frac{\la y|\Pi|y\ra}{\la
x|\Pi|x\ra}}. \ee
\end{lemma}
\begin{proof}
Statement~(1) can be
proved by contradiction. Assume $x\in \calS(\psi)$ and $\la x|\Pi|x\ra=0$.
Then $\Pi\, |x\ra=0$ and thus $\la x|\psi\ra = \la
x|\Pi|\psi\ra = 0$ which contradicts the definition of $\calS(\psi)$.
To prove the first part of statement~(2) note that
$\la x|\Pi|y\ra>0$ implies
\[
\la y|\psi\ra = \la y|\Pi|\psi\ra \ge \la y|\Pi|x\ra \la x|\psi\ra
>0,
\]
that is, $y\in \calS(\psi)$. The identity Eq.~(\ref{ratio}) follows
from Lemma~\ref{lemma:key}. Indeed, consider a decomposition of
$\Pi$ into non-negative pairwise orthogonal rank-one projectors:
\[
\Pi=\sum_{a=1}^{q} |\psi_a\ra\la \psi_a|, \quad
q=\mathrm{rank}(\Pi).
\]
The condition $\la x|\Pi|y\ra>0$ implies that $x$ and $y$ belong to
the same rank-one block of $\Pi$, that is
\bea \Pi\, |x\ra &=& \la
\psi_a|x\ra \, |\psi_a\ra = \sqrt{\la x|\Pi|x\ra}\,
|\psi_a\ra, \nn \\
\Pi\, |y\ra &=& \la \psi_a|y\ra \, |\psi_a\ra =
 \sqrt{\la y|\Pi|y\ra}\,
|\psi_a\ra \nn \eea for some block $a$. Now we have
\bea
\la x|\psi\ra &=& \la x|\Pi|\psi\ra =\sqrt{\la x|\Pi|x\ra}\,
\la \psi_a|\psi\ra, \nn \\
\la y|\psi\ra &=& \la y|\Pi|\psi\ra =\sqrt{\la y|\Pi|y\ra}\, \la
\psi_a|\psi\ra. \nn \eea Computing the ratio $\la y|\psi\ra /\la
x|\psi\ra$ we get Eq.~(\ref{ratio}).
\end{proof}

Now we are ready to prove Eq.~(\ref{amplitude|ratio}).
\begin{lemma}
\label{lemma:amplitude|ratio}
Let $H=\sum_{a=1}^M H_a$ be some SFF Hamiltonian and $\Pi_a$ be
the spectral projector corresponding to the zero eigenvalue of $H_a$.
Suppose $H\, |\psi\ra=0$ for some non-negative state $|\psi\ra$.
If for some  $x\in \calS(\psi)$, $y\in \Sigma^n$, and
$a\in \{1,\ldots,M\}$  one has
$\la y|H_a|x\ra<0$ then $y\in \calS(\psi)$ and
\[
\frac{\la y|\psi\ra}{\la x|\psi\ra} =\sqrt{ \frac{\la y|\Pi_a|y\ra}{\la x|\Pi_a|x\ra}}.
\]
\end{lemma}
\begin{proof}
Without loss of generality $x\ne y$.
Let us show that $\la y|H_a|x\ra<0$ implies $\la y|\Pi_a|x\ra>0$.
Indeed, let $\delta>0$ be the second smallest eigenvalue of $H_a$.
Then $0\le \delta (I-\Pi_a)\le H_a$. Define a Hermitian operator $O=|y\ra\la x|+|x\ra\la y|$.
It follows that $-\delta \trace{(O \Pi_a)}=\delta \trace{(O(I-\Pi_a))}  \le \trace{(O H_a)} = 2\la x|H_a|y\ra$.
Since we already know that $\Pi_a$ has real matrix elements, see Proposition~\ref{prop:SFF},
we get $\delta \la x|\Pi_a|y\ra \ge -\la x|H_a|y\ra \ge 0$ and thus $\la x|\Pi_a|y\ra>0$.
Now the lemma follows from Lemma~\ref{lemma:key}.
\end{proof}

\section{Simulation of the adiabatic evolution}
\label{sec:adiabatic}

In this section we prove Theorem~\ref{thm:adiabatic}. Let us start
with discretizing the adiabatic evolution. Define $H^{(j)}=H(j/T)$,
$j=0,\ldots,T$ where $T$ is a large integer that will be chosen
later. Using the bound $\| dH(s)/d(s)\|\le J$ we get \be
\label{max|increment} \| H^{(j+1)} - H^{(j)}\| \le \frac{J}{T}. \ee
The next step is to bound the overlap (inner product) between the
instantaneous ground-states at time $j$ and $j+1$. Let
$|\psi^{(j)}\ra$ be the ground-state of $H^{(j)}$, that is,
$H^{(j)}\, |\psi^{(j)}\ra=0$. We can assume that $|\psi^{(j)}\ra$
are non-negative states, see Corollary~\ref{cor:ground|states}.
\begin{lemma}
\label{lemma:overlap}
Let $\Delta$ be the smallest spectral gap of $H^{(j)}$, $j=0,\ldots,T$. Then for any $j$ one has
\be
\la \psi^{(j+1)}|\psi^{(j)}\ra \ge 1- \frac{J^2}{T^2 \Delta^2}.
\ee
\end{lemma}
A more general version of this lemma was proved in~\cite{Aharonov:03}. For the sake of completeness
we prove the lemma below.
\begin{proof}
Consider a decomposition $|\psi^{(j+1)}\ra=a|\psi^{(j)}\ra + b|\psi^{(j)}_\perp\ra$,
where $|\psi^{(j)}_\perp\ra$ is a normalized  vector orthogonal to  $|\psi^{(j)}\ra$,
so that $|a|^2+|b|^2=1$. Then
\[
\|\, H^{(j)} |\psi^{(j+1)}\ra\, \| = |b|\cdot  \| \, H^{(j)}\, |\psi^{(j)}_\perp\ra \, \| \ge |b| \Delta.
\]
On the other hand,
\[
\|\, H^{(j)} |\psi^{(j+1)}\ra\, \| = \|\, (H^{(j)} - H^{(j+1)} ) |\psi^{(j+1)}\ra\, \| \le \| H^{(j)} - H^{(j+1)} \|.
\]
Taking into account the bound Eq.~(\ref{max|increment}) we arrive at $|b|\le J\Delta^{-1} T^{-1}$.
Therefore
\[
\la \psi^{(j+1)}|\psi^{(j)}\ra =a \ge a^2 =1-|b|^2 \ge 1-J^2 \Delta^{-2} T^{-2}.
\]
\end{proof}

For every $j=0,\ldots,T$ define an operator
\[
G^{(j)}=I-\beta H^{(j)}, \quad \mbox{where} \quad
\beta^{-1} =\max_{0\le s\le 1} \sum_{a=1}^M \|H_a(s)\| .
\]
By definition of a local Hamiltonian
$\|H_a(s)\|\le poly(n)$, $M\le poly(n)$ and thus $\beta \ge poly(1/n)$.
The following properties of $G^{(j)}$ follow directly from the definition.
\begin{proposition}
The operator $G^{(j)}$ has a non-negative matrix in the standard
basis. The spectrum of $G^{(j)}$ belongs to the interval $[0,1]$ and
$|\psi^{(j)}\ra$ is the only eigenvector of $G^{(j)}$ with
eigenvalue $1$. The eigenvalue $1$ is separated from the rest of the
spectrum by a gap which is at least $\beta \Delta$.
\end{proposition}

For any $x,y\in \calS(\psi^{(j)})$ define the transition probability
\be P^{(j)}_{x\to y} = \frac{\la y|\psi^{(j)}\ra}{\la
x|\psi^{(j)}\ra} \, \la y|G^{(j)}|x\ra. \ee As was explained in
Section~\ref{subs:technique|adiabatic}, $P^{(j)}$ defines a random
walk on the set $S(\psi^{(j)})$ with the stationary distribution
\[
\pi^{(j)}_x=\la x|\psi^{(j)}\ra^2.
\]
Since $P^{(j)}$ is obtained from $G^{(j)}$ by a similarity transformation,
the spectrum of the matrix $P^{(j)}$ coincides with the spectrum of $G^{(j)}$
restricted to the subspace spanned by basis vectors from $S(\psi^{(j)})$.
Since the largest eigenvector of $G^{(j)}$ belongs to this subspace, we conclude that
$P^{(j)}$ has a spectral gap at least $\beta \Delta$.
Lemma~\ref{lemma:overlap} allows one to bound the fidelity between the stationary
distributions $\pi^{(j)}$ and $\pi^{(j+1)}$,
\be
\label{fidelity}
F(\pi^{(j)},\pi^{(j+1)})=\sum_{x\in \Sigma^n} \sqrt{\pi^{(j)}_x \pi^{(j+1)}_x } = \la \psi^{(j)} |\psi^{(j+1)}\ra \ge 1-\frac{J^2}{T^2 \Delta^2}.
\ee
In order to simulate the adiabatic evolution we shall generate a sequence of strings
$x^{(0)},x^{(1)},\ldots,x^{(T+1)}\in \Sigma^n$ such that $x^{(0)}$ is an arbitrary string satisfying
$\la x^{(0)} | \psi^{(0)}\ra  \ge 2^{-n}$, and $x^{(j+1)}$ is generated from $x^{(j)}$ by
making $L$ steps of the random walk $P^{(j)}$ starting from $x^{(j)}$.
 We shall try to choose  the number of steps $L$  such that for all $j$ the distribution
$\pi^{(j)}$ has a non-negligible probability at $x^{(j)}$. More specifically, we
want the following inequality to be satisfied with high probability for all $j=0,1,\ldots,T$:
\be
\label{warm|start}
\la x^{(j)} | \psi^{(j)}  \ra \ge 2^{-n-2}.
\ee
A string $x^{(j)}$ satisfying Eq.~(\ref{warm|start}) will be referred to as a {\it warm start} (for the random walk
$P^{(j)}$).
Let $\tilde{\pi}^{(j)}_x$ be the probability distribution of a string $x$
obtained by making  $L$ steps of $P^{(j)}$ with
a fixed warm start $x^{(j)}$.
Using the definition of the random walk $P^{(j)}$
one can express the statistical difference between
the distributions $\tilde{\pi}^{(j)}$ and $\pi^{(j)}$ as
\[
\| \tilde{\pi}^{(j)} - \pi^{(j)} \|_1  = \frac1{2 \la x^{(j)} |\psi^{(j)}\ra} \sum_{x\in \Sigma^n}
\la x|\psi^{(j)} \ra  \left| \la x|(\tilde{G}^{(j)})^L |x^{(j)}\ra \right|
\]
where $\tilde{G}^{(j)}=G^{(j)}-|\psi^{(j)}\ra\la \psi^{(j)}|$. Applying the Cauchy-Schwartz inequality and taking into
account that  the largest eigenvalue of $\tilde{G}^{(j)}$ is at most $1-\beta\Delta$, we arrive at
\[
\| \tilde{\pi}^{(j)} - \pi^{(j)} \|_1  \le \frac1{2 \la x^{(j)} |\psi^{(j)}\ra}  \| \, (\tilde{G}^{(j)})^L |x^{(j)}\ra \, \|
\le 2^{n+1} (1-\beta \Delta)^L.
\]
Clearly the statistical difference can be made exponentially small with $L=poly(n)$.
Neglecting  exponentially small errors,
we shall assume for simplicity that $\tilde{\pi}^{(j)}=\pi^{(j)}$, that is, given  the warm start
condition at step $j$,
the endpoint $x^{(j+1)}$ of the walk $P^{(j)}$ is drawn from the stationary distribution $\pi^{(j)}$.
All that remains is to evaluate the probability for the warm start condition to be violated.
In order to achieve this, let us introduce the notion of $t$-balanced strings.

\begin{definition}
Let $\pi$, $\rho$ be probability distributions on $\Sigma^n$ and $t\ge 1$ be a real number.
A string $x\in \Sigma^n$ is called {\it $t$-balanced} with respect to $\pi$ and $\rho$ iff
$\pi_x>0$, $\rho_x>0$ and $t^{-1}\le \pi_x/\rho_x \le t$.
\end{definition}
We shall denote a set of all $t$-balanced strings as $M_t(\pi,\rho)$, that is
\be
\label{t-strings}
M_t(\pi,\rho)=\{ x\in \Sigma^n\, : \, \pi_x>0, \quad \rho_x>0, \quad t^{-1}\le \frac{\pi_x}{\rho_x} \le t\}.
\ee
Let $F(\pi,\rho)=\sum_{x\in \Sigma^n} \sqrt{\pi_x\rho_x}$ be the fidelity between $\pi$ and $\rho$.
\begin{lemma}
\label{lemma:t-balanced}
Suppose $F(\pi,\rho)\ge 1-\delta$ and $t\ge 4$. Then the probability for a string $x$ drawn from the distribution $\pi$ to be
$t$-balanced is
\be
\label{t-strings|are|likely}
\sum_{x\in M_t(\pi,\rho)} \pi_x \ge 1-\frac{2\delta}{1-2t^{-1/2}}.
\ee
\end{lemma}
\begin{proof}
Indeed, if $x$ is not $t$-balanced then
\[
\sqrt{\pi_x\rho_x} \le t^{-1/2} \max{\{\pi_x,\rho_x\}}\le t^{-1/2} (\pi_x+\rho_x).
\]
Thus
\[
1-\delta \le \sum_{x\in M_t(\pi,\rho)} \sqrt{\pi_x \rho_x} + t^{-1/2} \sum_{x\notin M_t(\pi,\rho)} (\pi_x + \rho_x).
\]
Taking into account that $\sqrt{\pi_x \rho_x}\le (1/2)(\pi_x+\rho_x)$ we get
\[
1-\delta \le
\frac12 \sum_{x\in M_t(\pi,\rho)} (\pi_x+\rho_x) + t^{-1/2} \sum_{x\notin M_t(\pi,\rho)} (\pi_x + \rho_x)=
 1 - (1/2 - t^{-1/2}) \sum_{x\notin M_t(\pi,\rho)} (\pi_x+\rho_x).
\]
It follows that
\[
\delta \ge (1/2 - t^{-1/2}) \sum_{x\notin M_t(\pi,\rho)} (\pi_x+\rho_x)
\]
which yields Eq.~(\ref{t-strings|are|likely}).

\end{proof}
Using Lemma~\ref{lemma:overlap} we can bound the probability for
$x^{(j+1)}$ to be $t$-balanced with respect to $\pi^{(j)}$ and
$\pi^{(j+1)}$ for $t=16$ as \be \label{balance} \prob{\left[
x^{(j+1)} \in M_{16}(\pi^{(j)},\pi^{(j+1)}) \right]} \ge
1-\frac{4J^2}{\Delta^2 T^2}. \ee Thus in order to prove that
$x^{(j+1)}$ is a warm start for $P^{(j+1)}$ it suffices to show that
$\pi^{(j)}$ has large enough probability at $x^{(j+1)}$. Recall that
$x^{(j+1)}$ is random string drawn from the stationary distribution
$\pi^{(j)}$ (with exponentially small error).
\begin{proposition}
Let $\pi$ be a probability distribution on $\Sigma^n$. The probability for a string $x$ drawn from
$\pi$ to satisfy $\pi_x\ge 2^{-2n}$ is at least $1-2^{-n}$,
\end{proposition}
\begin{proof}
Indeed,
\[
\sum_{x\, : \, \pi_x\le 2^{-2n}} \pi_x \le 2^n 2^{-2n} =2^{-n}.
\]
\end{proof}

Again ignoring, for simplicity, events that occur with exponentially
small probability, we assume that $\pi^{(j)}_{x^{(j+1)}} \ge
2^{-2n}$ and thus  Eq.~(\ref{balance})  implies \be \label{balance1}
\prob{\left[ \la x^{(j+1)} | \psi^{(j+1)}\ra \ge 2^{-n-2}  \right]}
\ge 1-\frac{4J^2}{\Delta^2 T^2}.
 \ee
Thus the conditional probability for $x^{(j+1)}$ to be a warm start provided that  $x^{(j)}$ is a warm start
is at least $1-4J^2 \Delta^{-2} T^{-2}$. Therefore the probability for $x^{(j)}$ to be a warm start
for all $j=1,2,\ldots,T$ is at least $(1-4J^2 \Delta^{-2} T^{-2})^T\approx 1-4J^2 \Delta^{-2} T^{-1}$
for $T\gg \Delta^{-2} J^2$. If the warm start condition has been violated at some step, we can not
guarantee that $x^{(T+1)}$ is drawn from the distribution $\pi^{(T)}$ (in fact, the final distribution
might not be defined in this case because the endpoint of some walk $P^{(j)}$
may be outside of $S(\psi^{(j+1)})$). Therefore, our simulation scheme
fails with probability at most  $\delta\sim J^2 \Delta^{-2} T^{-1}$.
It follows that $T=poly(n,\delta^{-1})$.

\section{Stoquastic $k$-SAT is contained in MA}
\label{sec:inMA}

Recall that we consider an instance of stoquastic $k$-SAT $(n,\calC,\epsilon)$,
where $\epsilon=n^{-O(1)}$ is a precision parameter  and $\calC=\{ H_a\}_{a=1,\ldots,M}$ is a stoquastic
system of $(n,k)$-constraints.

A formal description of the prover's strategy is the following.
Consider first a  {\it yes-instance}. Let $|\psi\ra\in \calQ^n$ be
any satisfying assignment (which can always be chosen as a
non-negative vector). The prover sends the verifier a string
 $w\in \Sigma^n$ corresponding to the largest amplitude of $|\psi\ra$, that is,
$\la w|\psi\ra \ge \la x|\psi\ra$  for all $x\in \Sigma^n$. In case of a {\it no-instance} the
prover may send the verifier an arbitrary string $w\in \Sigma^n$.

The verifier starts by choosing $\beta>0$ such that $\beta \| H_a\|
\le 1$ for all $a$ and an integer $L$ such that \be 2^{\frac{n}{2}}
\, (1-\epsilon \beta M^{-1})^L\le \frac13. \ee Note that this
inequality can be satisfied with $L=poly(n)$. Then the verifier
performs a random walk with $L$ steps as prescribed below.
\begin{center}
\fbox{%
\parbox{12cm}{
{\bf Step~1:} Receive a string $w\in \Sigma^n$  from the prover. Set $x_0=w$.\\ \\
{\bf Step~2:} \parbox[t]{10cm}{Suppose the current state of the walk
is $x_j$. Verify that $x_j\in S_{good}$, see Eq.~(\ref{good|bad}). Otherwise,  output `no'.}\\ \\
{\bf Step~3:} If $j=L$ goto Step~8.\\ \\
{\bf Step~4:} Generate a random uniform $a\in \{1,\ldots,m\}$.\\  \\
{\bf Step~5:} Find the set $\calN_a(x_j)=\{y\in \Sigma^n\, : \, \la x_j|\Pi_a|y\ra>0\}$.\\  \\
{\bf Step~6:} \parbox[t]{12cm}{Generate a random $x_{j+1}\in \calN_a(x_j)$ from the distribution
\be
\label{P^a}
P^a_{x_j \to x_{j+1}} = \sqrt{ \frac{\la  x_{j+1}|\Pi_a|x_{j+1}\ra}{\la x_j| \Pi_a |x_j\ra}} \, \la x_{j+1} |G_a|x_j\ra,
\quad G_a=I-\beta H_a.
\ee
}\\ \\
{\bf Step~7:} \parbox[t]{12cm}{Compute and store the number
\be\label{rj} r_{j+1} = \frac{P_{x_j\to x_{j+1}}^a}{\la
x_{j+1}|G_a|x_{j}\ra}. \ee
Set $j\to j+1$ and goto Step~2.}\\ \\
{\bf Step~8:} Verify that $\prod_{j=1}^L r_j \le 1$. Otherwise output `no'.\\ \\
{\bf Step~9:} Output `yes'.} } \label{prot:MA}
\end{center}

\vspace{0.3cm}

\subsection{Completeness of the protocol} \label{subs:yesinstance}
In this subsection we prove that for yes-instances the verifier
outputs `yes' with probability $1$. Let $|\psi\ra$ be a satisfying
assignment chosen by the prover. We can assume that $|\psi\ra$ is a
non-negative vector, see Corollary~\ref{cor:ground|states}.
Lemma~\ref{lemma:amplitude|ratio} implies that for any choice of $a$
at Step~4 one has
\[
P^a_{x_j\to x_{j+1}} = \frac{\la x_{j+1}|\psi\ra}{\la x_j|\psi\ra} \la x_{j+1} |G_a|x_j\ra
\]
and thus the overall probability of a transition from $x_j$ to $x_{j+1}$ at Steps~4,5,6 is
\be
\label{G}
P_{x_j\to x_{j+1}} = \frac{\la x_{j+1}|\psi\ra}{\la x_j|\psi\ra} \la x_{j+1} |G|x_j\ra, \quad G=\frac1M \sum_{a=1}^M G_a.
\ee
It follows from Lemma~\ref{lemma:amplitude|ratio} that the walk never leaves the set $S(\psi)$.
Using the above expression for $P^a_{x_j\to x_{j+1}}$  one has
\[
r_j=\frac{\la x_{j+1} |\psi\ra }{\la x_j|\psi\ra}.
\]
Therefore,
\[
\prod_{j=1}^L r_j=  \frac{\la x_{L} |\psi\ra }{\la x_0|\psi\ra} =
\frac{\la x_{L} |\psi\ra }{\la w|\psi\ra}.
\]
Taking into account that $w$
is a sting with the largest amplitude  one can see that
$\prod_{j=1}^L r_j\le 1$ for all possible $x_L\in S(\psi)$ and thus
the test at Step~8 will always be passed. Thus the verifier
outputs `yes' with probability $1$.

Let us remark that the completeness of the protocol is not affected
by the precision up to which the verifier approximates the
probability distribution $P_{x_j\to x_{j+1}}^a$ at Step~6 as long as
$x_{j+1}\in \calN_a(x_j)$ with probability $1$ and the coefficients
$r_j$ are computed using {\it exact} formulas
Eqs.~(\ref{P^a},\ref{rj}). The precision will be important for the
soundness of the protocol, see the next section.

\subsection{Soundness of the protocol} \label{subs:noinstance} In
this subsection we shall prove that for no-instances the verifier
outputs `yes' with probability at most $1/3$. Without loss of
generality the witness $w=x_0$ is a good string (otherwise the
verifier outputs `no' at the very first step of the walk). The
probability for the walk starting from $x_0\in S_{good}$ to stay in
$S_{good}$ at every step $j=1,2,\ldots, L$ is
\[
p_{good}(L)= \frac1{M^L} \quad \sum_{a_1,\ldots,a_L=1,\ldots,M} \quad \sum_{x_1,\ldots,x_{L}\in S_{good}} P_{x_0\to x_1}^{a_1}\,
 P_{x_1\to x_2}^{a_2} \cdots P_{x_{L-1}\to x_{L}}^{a_L}.
\]
Taking into account Eq.~(\ref{rj}) one gets
\[
p_{good}(L)=\frac1{M^L}
\sum_{\ba{ll} {\scriptstyle x_1,\ldots,x_{L}\in S_{good}} \\
{\scriptstyle a_1,\ldots,a_L=1,\ldots,M }\\
\ea}   \left(\prod_{j=1}^L r_j\right)
\la x_0|G_{a_1}|x_1\ra \, \la x_1|G_{a_2}|x_2\ra \cdots \la x_{L-1}|G_{a_L}|x_L\ra.
\]
Here  the coefficients $r_j$ are functions of the ``trajectory"
$x_0,\ldots,x_L$ and $a_1,\ldots,a_L$ of the walk.
At this point we invoke the test at Step~8. The verifier
outputs `yes' iff the walk stays in $S_{good}$ at every step
$j=1,\ldots,L$ {\em and} $\prod_{j=1}^L r_j\le 1$. Thus the
probability $p_{yes}(x_0)$ for the verifier to output `yes' for a fixed starting string $x_0=w\in
S_{good}$ can be bounded from above as
\[
p_{yes}(x_0) \le
\frac1{M^L}
\sum_{\ba{ll} {\scriptstyle x_1,\ldots,x_{L}\in S_{good}} \\
{\scriptstyle a_1,\ldots,a_L=1,\ldots,M }\\
\ea}
\la x_0|G_{a_1}|x_1\ra \, \la x_1|G_{a_2}|x_2\ra \cdots \la x_{L-1}|G_{a_L}|x_L\ra.
\]
Using the operator $G=M^{-1} \sum_{a=1}^M G_a = I-\beta M^{-1} H$ we have
\[
p_{yes}(x_0) \le
\sum_{x_1,\ldots,x_{L}\in S_{good}}
\la x_0|G|x_1\ra \, \la x_1|G|x_2\ra \cdots \la x_{L-1}|G|x_L\ra.
\]
Taking into account that all matrix elements of $G$ are
non-negative, we get
\[
p_{yes}(x_0)
 \le
 2^{\frac{n}2} \la x_0|G^L|+\ra,
\]
where $|+\ra = 2^{-n/2} \sum_{x\in \Sigma^n} |x\ra$ is the uniform
superposition of all $2^n$ basis vectors.
Let $\lambda$ be the largest eigenvalue of $G$.
The promise for a no-instance implies that
 $\lambda\le  1-\epsilon \beta M^{-1}$ and thus
\[
\la x_0|G^L|+\ra \le \lambda^L \le (1-\epsilon \beta M^{-1})^L.
\]
Therefore
\[
p_{yes}(x_0)\le
2^{\frac{n}2} (1-\epsilon \beta M^{-1})^L \le \frac13
\]
for any starting string $x_0$. It proves that the verifier outputs
'yes' with probability at most $1/3$. Now suppose that Step~6 is
implemented with some finite precision using a probability
distribution $\tilde{P}_{x_j\to y}^a$ such that
 \be
\sum_{y\in \calN_a(x_j)}
\left| \tilde{P}_{x_j\to y}^a - P_{x_j\to y}^a \right| \le \delta \quad
\mbox{for any} \quad x_j\in S_{good}, \quad \mbox{for any $a=1,\ldots,m$}. \label{approxP}
\ee
One can
easily verify that Eq.~(\ref{approxP}) implies
\[
\left| \sum_{x_1,\ldots,x_L\in S_{good}} P_{x_0\to x_1}^{a_1} \, P_{x_1\to
x_2}^{a_2} \cdots P_{x_{L-1}\to x_L}^{a_L} - \tilde{P}_{x_0\to x_1}^{a_1} \, \tilde{P}_{x_1\to x_2}^{a_2}
\cdots \tilde{P}_{x_{L-1}\to x_L}^{a_L} \right| \le L\delta.
\]
Thus using an approximate probability distribution at Step~6 leads
to corrections of order $L\delta$ to the overall acceptance
probability. Choosing $\delta\ll L^{-1}$ we can get an acceptance
probability smaller than  $1/2$ which can be amplified to $1/3$
using  standard majority voting.

\subsection{Simplified stoquastic $k$-SAT}
\label{subs:ustoq}

Let $\calC=\{H_a=I-\Pi_a\}_{a=1,\ldots M}$ be a system of
$(n,k)$-constraints where each $\Pi_a$ is a projector with matrix
elements $0,1/2,1$. Recall that verifying satisfiability of $\calC$
is  a problem complete for Promise-MA if $k\ge 6$, see
Theorem~\ref{thm:MA}. It is of interest to consider in more detail
how the random walk algorithm described above works for this
simplified version of stoquastic $k$-SAT.

The identity $\Pi_a=\Pi_a^2$ implies that $\Pi_a$ is a
block-diagonal matrix (up to a permutation of basis vectors)
such that every one-dimensional block is either $0$ or $1$,
every two-dimensional block is a matrix
\[
\frac12\, \left(\ba{cc} 1 & 1 \\ 1 & 1 \\ \ea\right),
\]
and there are no blocks with dimension higher than two. Let us
define a graph $G=(V,E)$ such that vertices of $G$ are $n$-bit
strings and a pair of strings $(x,y)$ is connected by an edge iff
there is at least one projector $\Pi_a$ such that $\langle x |\Pi_a
|y \rangle=1/2$. Note that $G$ has degree at most $M$, the number of
projectors $\Pi_a$. The vertices of $G$ can be partitioned into good
and bad vertices, see Eq.~(\ref{good|bad}). In other words,
\[
S_{good}=\{x\in \Sigma^n \, : \, \la x|\Pi_a|x\ra \in \{1,\frac12 \}
\quad \mbox{for all $a$}\},
\]
and $S_{bad}=\{x\in \Sigma^n \, : \, \la x|\Pi_a|x\ra=0 \quad \mbox{for some $a$}\}$.

The random walk described by Steps~4, 5 and~6 of the verifier's
protocol can now be simplified as follows (we shall assume that
$\beta=1$ so that $G_a=\Pi_a$).
 Suppose at some step $j$ the walk was at some good vertex $x_j$.
Then at the next step  the walker moves to one of the nearest
neighbors of $x_j$ with probability $1/(2M)$ and stays at the vertex
$x_j$ with probability $1-\mathrm{deg}(x_j)/2M$, where
$\mathrm{deg}(x_j)$ is the number of edges incident to $x_j$.

One can easily observe that the system $\calC$ is satisfiable, $\unsat{\calC}=0$,
iff the graph $G$ has a connected component $G'=(V',E')$ that contains only good vertices,
$V'\subseteq S_{good}$.  Given such a connected component, a state
$|\psi\ra=\sum_{u\in V'} |u\ra$ is a satisfying assignment, that is, $\Pi_a \, |\psi\ra=|\psi\ra$
for all $a$.
Hence for a `yes'-instance the prover can simply give any string in $V'$.
Note also that the test at Step~8 of the protocol is not needed now since $r_j=1$
for any realization of the walk.

The promise on the unsat-value for `no'-instances can be translated
into a promise that for any starting vertex the random walk
will hit a bad vertex after polynomially number of steps. Note that the classical SAT problem
corresponds to the special case in which the projectors have no
$1/2$ matrix-elements and hence there are no edges in the associated
graph (and hence also no walk).

\section{Discussion}
\label{sec:discussion}
We hope that the stoquastic $k$-SAT problem
may potentially lead to new insights into the question
whether ${\rm MA} \stackrel{?}{=} {\rm NP}$. Note that it is widely
believed that MA=NP, see e.g.~\cite{Santhanam07}.
Indeed, the simplified stoquastic $k$-SAT problem described in Section~\ref{subs:ustoq}
is the problem
of deciding whether the associated graph $G$ with $2^n$ vertices has a
connected component $G'$ that contains only good vertices (the verifier can efficiently check whether a vertex
is good).
 If the size of $G'$ were polynomial, then a NP proof-system would suffice,
since a prover can simply list all
vertices in $G'$. However in general the size of $G'$
 could be exponentially large and there is no time
or space to explore or list the whole subgraph $G'$, hence the need for
randomness. For `yes'-instances the prover simply gives the verifier a vertex $v\in G'$
and a random walk on $G$ starting from $v$ will always stay in $G'$. For
`no'-instances, the promise guarantees that no matter where one
starts the random walk, with a polynomial number of steps one will
always hit a bad vertex with high probability. The derandomization
question is the question whether a pseudo-random walk from any
starting vertex using a random bit string of length $O(\log n)$ will
also hit a bad vertex with sufficiently high probability.

Another open question is the simulatibility of stoquastic adiabatic
computation in general.
Given a stoquastic Hamiltonian $H$ and its ground state $|\psi\ra$
satisfying $H\, |\psi\ra=0$ (which can be always achieved by an energy shift)
one can still use Eq.~(\ref{transition|matrix})  to define a random walk on the set $\calS(\psi)$.
However, in order to simulate this random walk on a classical computer
 one must be able to compute the ratio of the amplitudes
$\la y|\psi\ra/\la x|\psi\ra$ for which no efficient algorithm is known.
Another possibility to define a random walk is to modify  the verifier's protocol
in Section \ref{sec:inMA}.
It suffices to modify Step~6 such that the walk stays at $x_j$ with
probability $1$ whenever $\la x_j|\Pi_a|x_j\ra=0$. An open
question is whether the stationary distribution of this modified
walk   is anyhow related to the distribution $\pi(x)=\langle
x|\psi\rangle^2$ associated with a ground-state $|\psi\ra$
of the stoquastic Hamiltonian.

\section*{Acknowledgments}
We would like to thank David DiVincenzo for useful comments and
discussions. This work was supported by NSA and ARDA through ARO
contract number W911NF-04-C-0098.

\section*{Appendix A}
The complexity class MA (Merlin-Arthur games) was introduced by Babai in~\cite{babai}.
A language $L$ is in MA iff there
exists a probabilistic polynomial-time machine $V$ (a verifier) that takes as input a pair $(x,w)$ where $x$ is
string representing an instance of
the problem, $w$ is a witness string, and such that
\bea
x\in L &\Rightarrow & \mbox{$V$ accepts $(x,w)$ with probability $1$ for some witness $w$}.\nn \\
x\notin L &  \Rightarrow & \mbox{$V$ accepts $(x,w)$ with
probability at most $1/3$ for any witness $w$}. \nn \eea

One gets an equivalent definition if the acceptance probabilities $1$ and $1/3$ are replaced by
$p_{yes}$ and $p_{no}$ such that $p_{yes}-p_{no}\ge 1/poly(n)$, see~\cite{furer89completeness}.
To better understand the relationship between MA and other complexity classes it is desirable to have some MA-complete problems, or at least, some
problems in MA that are not known to be in NP. Unfortunately,
no such problems are currently known.
This lack of interesting problems in MA is not very surprising though
because probabilistic algorithms are usually allowed to give an inconclusive answer for some
inputs (the acceptance probability is close to $1/2$) while the definition above does not allow that.
From this perspective it is more natural to define MA as a class of {\em promise problems}.
Recall that a promise problem is a pair of non-intersecting sets $L_{\yes},L_{\no} \subseteq \{0,1\}^*$
that represent yes-instances and no-instances respectively.
\begin{definition}
A promise problem $(L_{\yes},L_{\no})$ belongs to MA iff
there exists a probabilistic polynomial-time machine $V$  taking as input a pair of
strings $(x,w)$ such that
\bea
\label{MApromise}
x\in L_{\yes} &\Rightarrow & \mbox{$V$ accepts $(x,w)$ with probability $1$ for some witness $w$}.\nn \\
x\in L_{\no} &  \Rightarrow & \mbox{$V$ accepts $(x,w)$ with
probability at most $1/3$ for any witness $w$}.\nn \eea
A promise problem $(L_{yes},L_{no})$ is MA-complete iff for any
promise problem $(L_{yes}',L_{no}')$ in MA there exists a function
$f\,: \, \{0,1\}^* \to \{0,1\}^*$ computable by a deterministic
polynomial-time machine  such that $f(L_{yes}')\subseteq L_{yes}$
and  $f(L_{no}')\subseteq L_{no}$.
\end{definition}

Note that the behavior of $V$ on instances  $x\notin L_{\yes}\cup
L_{no}$ may be completely arbitrary. Throughout this paper MA refers
to the class of promise problems rather than the class of languages.
It should be mentioned  that many important results concerning NP
are formulated in terms of promise problems, for example, the
inapproximability version of the PCP theorem~\cite{Dinur06}, or the
complexity of $k$-SAT with a unique solution~\cite{VV86} (see the
survey~\cite{Goldreich_survey} for other examples). Thus one can
also expect to get more insight in the complexity of MA by studying
promise problems.

\section*{Appendix B}
This section serves three purposes. First, we  prove
Proposition~\ref{prop:adiabatic1}. Secondly, we prove the second
part of Theorem~\ref{thm:MA}, that is, MA-hardness of stoquastic
$6$-SAT. Thirdly, we explain in more details the connection between
SFF Hamiltonians and classical probabilistic computation mentioned
in Section~\ref{subs:coherent}. All these results follow directly
from the clock Hamiltonian construction of~\cite{KSV:computation}
and the analysis performed in~\cite{ADKLLR:04,BDOT:06}.

We start form reviewing the clock Hamiltonian construction.
Let $U=U_L \cdots U_2 U_1$ be a quantum circuit acting on $N$ {\it data qubits}
with $L=poly(N)$ gates. We assume that the $N$ data qubits are partitioned into
two groups: $N_a$ ancillary qubits and $N_w$ witness qubits (one may have $N_w=0$).
Each ancillary qubit $k$ is initialized by some pure state $|\phi_k\ra$.
The witness qubits may be initialized by an arbitrary pure state $|\psi^{wit}\ra\in \calQ^{N_w}$.
Accordingly, the input state of the circuit is $|\psi_{in}\ra=|\psi^{anc}\ra\otimes |\psi^{wit}\ra$, where
$|\psi^{anc}\ra$ is a tensor product of the ancillary states $|\phi_k\ra$.
Let $|\psi_j\ra=U_j\cdots U_1\, |\psi^{anc}\ra \otimes |\psi^{wit}\ra$, $j=0,\ldots,L$,
be a state obtained by terminating the circuit after the $j$-th gate.
We adopt a convention that $|\psi_0\ra=|\psi_{in}\ra$ is the input state.
The output state of the circuit is $|\psi_L\ra$.

Consider a composite system that consists of $N$ data qubits and $L+1$ {\it clock qubits}.
Let $|j\ra_u=|1^{j+1} 0^{L-j}\ra\in \calQ^{L+1}$ be the unary encoding
of the time steps $j=0,\ldots,L$.
Define a linear subspace $\calH\subseteq \calQ^{N+L+1}$ as
\bea
\calH&=&\left\{ |\phi\ra=\sum_{j=0}^L |\psi_j\ra \otimes |j\ra_u, \quad
|\psi_j\ra=U_j\cdots U_1\, |\psi^{anc}\ra \otimes |\psi^{wit}\ra, \quad |\psi^{wit}\ra\in
\calQ^{N_w} \right\} \nn \\
&&    \label{eq:history}
\eea
States from $\calH$ represent computational paths of the verifier's quantum computer
starting from an arbitrary witness state $|\psi^{wit}\ra$.
We shall label the $j$-th clock qubit as $cl(j)$,
$j=0,\ldots,L$. Note that the clock qubit $cl(0)$ is always set to
$1$. For any
$j=1,\ldots,L$,  the clock qubit $cl(j)$ is a flag telling whether the
gate $U_j$ has or has not been applied.
The $k$-th ancillary qubit  will be labeled $a(k)$, $k=1,\ldots,N_a$.
Let us show that $\calH$ is spanned by ground states of some SFF
Hamiltonian. Indeed, introduce $3$-qubit constraints
\be
\label{Hinit}
H^{init}_k =(I-|\phi_k\ra\la \phi_k|)_{a(k)} \otimes |10\ra\la
10|_{cl(0),cl(1)}, \quad j=k,\ldots,N_a,
\ee
States satisfying these constraints
(i.e. zero eigenvectors of $H^{init}_k$)  satisfy correct initial conditions.
Introduce also constraints
\bea
\label{Hprop}
H^{prop}_j &=&
\frac12 |1\ra\la 1|_{cl(j-1)} \otimes \left(\vphantom{U_j^\dag}
|1\ra\la 1|_{cl(j)} + |0\ra\la 0|_{cl(j)} \right. \nn \\
&& \left.  - |1\ra\la 0|_{cl(j)} \otimes U_j - |0\ra\la
1|_{cl(j)} \otimes U_j^\dag\right) \otimes
|0\ra\la 0|_{cl(j+1)}
\eea where $j=1,\ldots,L$. States
satisfying these constraints obey the correct propagation
rules relating computational states at different time steps.
Finally, introduce $2$-qubit constraints
\[
H^{clock}_0=|0\ra\la 0|_{cl(0)}, \quad H^{clock}_l = |01\ra\la 01|_{cl(l-1),cl(l)}, \quad l=1,\ldots,L.
\]
States satisfying these constraints belong to the subspace spanned by ``legal" clock states,
i.e., $\calQ^{N}\otimes |j\ra_c$, $j=0,\ldots,L$.
Therefore we arrive at
\[
\calH=\left\{|\phi\ra \in \calQ^{N+L+1}\, : \, H^{init}_k\, |\phi\ra = H^{prop}_j\, |\phi\ra =H^{clock}_l |\phi\ra=0 \quad \mbox{for all $j,k,l$} \right\}.
\]
Define a clock Hamiltonian \be \label{clock1} H=\sum_{k=1}^{N_a}
H_k^{init}  + \sum_{j=1}^L H_j^{prop} + \sum_{l=0}^L  H^{clock}_l.
\ee
It follows that $\calH$ is the ground-subspace of $H$. Note that all
terms in $H$ are positive semi-definite and any vector from $\calH$
is a zero eigenvector of $H$. Thus $H$ is a frustration-free
Hamiltonian.

\noindent
{\bf Proof of Proposition~\ref{prop:adiabatic1}.}
In the case of the adiabatic evolution there are no witness qubits, $N_w=0$.
Accordingly, the clock Hamiltonian has a  unique ground-state
\[
|\psi\ra = (L+1)^{-1/2} \sum_{j=0}^L |\psi_j\ra \otimes |j\ra_u.
\]
It was shown by many researchers, see for instance Lemma~3.11 in~\cite{ADKLLR:04}
and the improved estimates in~\cite{Ruskai:06}, that the spectral gap of the clock
Hamiltonian can be bounded as $\Delta=\Omega(1/L^2)=1/poly(N) $ regardless of the choice of the gates $U_1,\ldots,U_L$.
Let us define a family of quantum circuits $U(s)=U_L(s) \cdots
U_1(s)$ where $0\le s\le 1$ and $U_j(s)$ interpolates smoothly
between $U_j(0)=I$ and $U_j(1)=U_j$ (without loss of generality
$\det{(U_j)}=1$ in which case the possibility of such a smooth
interpolation follows from the connectivity of the special unitary
group).
 Applying  definition
Eq.~(\ref{clock1}) we get a family of frustration-free clock
Hamiltonians $H(s)$, $0\le s\le 1$ satisfying conditions
(A1),(A2),(A3). Without loss of generality the last $(1-\delta)L$
gates of the circuit $U$ are identity gates. Therefore the
ground-state of $H$ satisfies $\la \psi|\psi_L\otimes A\ra\ge
1-\delta$, where $|A\ra=(L+1)^{-1/2}\sum_{j=0}^L |j\ra_u$ is an
ancillary state. Thus $|\psi\ra$ approximates the output state
$|\psi_L\ra$ with precision $\delta$ (after discarding the ancilla
$|A\ra$).

\noindent{\bf Stoquastic $6$-SAT is MA-hard.} It was shown
in~\cite{BDOT:06} (see Lemma~2 in~\cite{BDOT:06}) that any classical
MA verifier $V$ can be transformed into a quantum verifier $V'$
which uses a quantum circuit $U$ involving only classical reversible
gates (for example, the $3$-qubit Toffoli gates) together with
ancillary states $|0\ra$, $|+\ra$, and measures one of the output
qubits in the $|0\ra,|1\ra$ basis. This transformation has a
property that the maximum acceptance probability of $V$ (over all
classical witnesses) is equal to the maximum acceptance probability
of $V'$ (over all quantum witnesses). In order to apply the clock
Hamiltonian construction to $V'$ we shall treat part of the
ancillary qubits as the qubits encoding an instance of the problem,
that is, we shall allow ancillas $|\phi_k\ra=|0\ra,|1\ra,|+\ra$. In
addition, we shall add one more term into $H$ representing the final
measurement. Define a $3$-qubit constraint
\[
H^{meas}=(I-\Pi_{out})\otimes |1\ra\la 1|_{cl(L)}.
\]
Here $\Pi_{out}$ is the projector used by $V'$  to decide whether he
accepts the witness (say, $\Pi_{out}$ projects the first data qubit
onto the state $|0\ra$). Note that $\calH$ contains a vector
satisfying $\Pi^{meas}$ iff the verifier $V'$ accepts some witness
state $|\psi^{wit}\ra$ with probability $1$. Define a system of
constraints \be \calC=\{ H^{init}_k, H^{prop}_j,
H^{clock}_l,H^{meas}\} \ee i.e. the system  including all the
constraints defined above. Since all ancillary states $|\phi_k\ra$
are either $|0\ra$ or $|+\ra$,  the off-diagonal matrix elements of
the operators $H^{init}_k$ are either $0$ or $-1/2$, see
Eq.~(\ref{Hinit}). Furthermore, since all gates $U_j$ are classical
Toffoli gates, the off-diagonal matrix elements of the operators
$H^{prop}_j$ are either $0$ or $-1/2$, see Eq.~(\ref{Hprop}).
 Finally, the operators  $H^{clock}$ and $H^{meas}$ are diagonal.
Thus $\calC$ is a stoquastic system of $(n,6)$-constraints where $n=N+L+1$.
By definition, the unsat-value of $\calC$ coincides with the smallest eigenvalue of
a clock Hamiltonian
\be
\label{clock2}
H'=H+ H^{meas},
\ee
where $H$ is defined in Eq.~(\ref{clock1}).
If  $V'$ accepts some witness state $|\psi^{wit}\ra$ with probability $1$ then
$\unsat{\calC}=0$. On the other hand, if $V'$ accepts any witness state with
probability at most $1-\epsilon$,
the derivation of~\cite{KSV:computation}
implies that the smallest eigenvalue of the Hamiltonian Eq.~(\ref{clock2})
can be bounded as
\[
\lambda_{min}(H')\ge c(1-\sqrt{1-\epsilon})L^{-3}\ge 1/poly(N).
\]
where $c$ is some positive constant. Thus in the latter case $\unsat{\calC}\ge 1/poly(N)$.
It follows that any problem in MA is reducible to the stoquastic $6$-SAT problem.

Finally,
we remark that any constraint from the system $\calC$ can be represented as
$I-\Pi$ where $\Pi$ is a non-negative projector with
matrix elements  belonging to the set $\{0,1,1/2\}$.

\noindent {\bf Coherent probabilistic computation.} The Hamiltonian
mentioned in Section~\ref{subs:coherent} is the clock Hamiltonian
defined in Eq.~(\ref{clock1}). The proof that the ground-state of
$H$ has the desired properties is completely analogous to the proof
of Proposition~\ref{prop:adiabatic1}.



\bibliographystyle{hunsrt}

\end{document}